\documentclass[runningheads,a4paper]{llncs}

\usepackage{amssymb}
\usepackage{graphicx}
\usepackage{url}
\usepackage{amsmath}
\usepackage{tikz}
\usetikzlibrary{calc}
\usetikzlibrary{arrows.meta,positioning}

\begin{document}

\mainmatter 

\title{Bit-counting complexity classes}

\titlerunning{Bit-counting complexity classes}

\author{ Tayfun Pay}

\authorrunning{Tayfun Pay}

\urldef{\mailsa}\path|tpay@gradcenter.cuny.edu|

\institute{{ \mailsa}}



%
%

\maketitle

\begin{abstract}

We define bit-counting complexity classes whose membership depends on the binary profile of the number of accepting paths of non-deterministic polynomial time Turing machines. We study the relationship between this new family of complexity classes and the classical complexity classes. We prove that the complexity class ${\bf PP}$ is contained in our comparison based bit-counting complexity classes ${\bf B_{|0|=|1|}P}$, ${\bf B_{|0|<|1|}P}$ and ${\bf B_{|0|>|1|}P}$. We then show that the comparison based bit-counting complexity classes and the complexity class ${\bf PP}$ are Turing equivalent, that is ${\bf P}^{\bf PP} = {\bf P}^{{\bf B_{|0|=|1|}P}}={\bf P}^{{\bf B_{|0|>|1|}P}}={\bf P}^{{\bf B_{|0|<|1|}P}}$. We then prove that the complexity classes ${\bf NP}$ and ${\bf CoNP}$ are contained in both of our parity based bit-counting complexity classes ${\bf B_{|0| \oplus}P}$ and ${\bf B_{|1| \oplus}P}$. We also show that the Turing closures of the parity based bit-counting complexity classes coincide, that is ${\bf P}^{{\bf B_{|0|\oplus}P}}={\bf P}^{{\bf B_{|1|\oplus}P}}$. We do this by proving that when either parity based bit-counting complexity class is provided as an oracle for a polynomial time Turing machine, then it can simulate the other one, that is ${\bf B_{|1| \oplus}P}\subseteq {\bf P}^{{\bf B_{|0| \oplus}P}}$ and ${\bf B_{|0| \oplus}P}\subseteq {\bf P}^{{\bf B_{|1| \oplus}P}}$.
\end{abstract}

\section{Introduction}

We introduce a family of bit-counting complexity classes defined by applying predicates to the count of 0's bits and the count of 1's bits in the binary representation of ${\bf \#P}$ \cite{V79} functions. Instead of comparing the number of accepting computation paths of a non-deterministic polynomial time Turing machine directly to a value, this family of complexity classes examines the number of 0's bits and the number of 1's bits that appear in the binary representation of the accepting path count. Fundamentally, our bit-counting complexity classes provide a different perspective on counting complexity such that rather than asking whether an accepting path count of a non-deterministic polynomial time Turing machine is small, large, zero, one, odd, even or exactly equal to a threshold, they ask whether the binary description of the count satisfies some bit-counting property. 

Let's demonstrate how our bit-counting complexity classes operate compared to the classical complexity classes with an example. Assume that we are given a Boolean formula $F$ in conjunctive normal form with 5 variables that has 27 satisfying truth assignments. We then know that the count 27 satisfies the numerical acceptance predicates associated with ${\bf PP}$ \cite{S75}, ${\bf CoC_=P}$, ${\bf CoMNS}$, ${\bf CoES}$, ${\bf \oplus P}$ \cite{PZ83}, ${\bf CoUS}$ as well as ${\bf NP}$, while it does not satisfy those associated with ${\bf C_=P}$ \cite{S75}, ${\bf MNS}$ \cite{CP18}, ${\bf ES}$ \cite{BHR00}, ${\bf US}$ \cite{BG82} and ${\bf CoNP}$. Yet, what we do with the bit-counting complexity classes is to examine the binary representation of 27 instead, which is 11011$_2$ and has one 0's bit and four 1's bits. We apply nine different bit-counting predicates to the binary representation of ${\bf \# P}$ functions that ask the following questions: 1) Are there no 1's bits? 2) Is there at least one 1's bit? 3) Are there no 0's bits? 4) Is there at least one 0's bit? 5) Are there equal number of 0's bits and 1's bits? 6) Are there more 0's bits than 1's bits? 7) Are there more 1's bits than 0's bits? 8) Is the number of 0's bits odd? and 9) Is the number of 1's bits odd? In our example, Boolean formula $F$ is a yes instance for questions 2), 4), 7) and 8), whereas it is a no instance for questions 1), 3), 5), 6) and 9). Essentially, these nine questions are the canonical complete problems for the bit-counting complexity classes that we study.

We first obtain equality among some classical complexity classes and bit-counting complexity classes purely based on their definitions. These are the bit-counting complexity classes that compare the count of either the 0's bits or the 1's bits to the value zero. These bit-counting complexity classes and their equivalent classical counterparts are as follows: 1)${\bf B_{|1|=0}P}={\bf CoNP}$, 2)${\bf B_{|1|>0}P}={\bf NP}$, 3)${\bf B_{|0|=0}P}={\bf MNS}$ and 4)${\bf B_{|0|>0}P}={\bf CoMNS}$.

We then consider comparison based bit-counting complexity classes that compare the 0's bits count and the 1's bits count as follows: 5)${\bf B_{|0|=|1|}P}$, 6)${\bf B_{|0|>|1|}P}$ and 7)${\bf B_{|0|<|1|}P}$. We show that the comparison based bit-counting complexity classes are strong enough to contain ${\bf PP}$, that is we prove that ${\bf PP}\subseteq{\bf B_{|0|=|1|}P}$, ${\bf PP}\subseteq{\bf B_{|0|>|1|}P}$ and ${\bf PP}\subseteq{\bf B_{|0|<|1|}P}$. We do this by using explicit ${\bf \#P}$ constructions that force the binary representation of the resulting count to have the desired balance between 0's bits and 1's bits. We also prove that these comparison based bit-counting complexity classes are contained in ${\bf P}^{\bf PP}$, that is we show that ${\bf B_{|0|=|1|}P}\subseteq{\bf P}^{\bf PP}$, ${\bf B_{|0|>|1|}P}\subseteq{\bf P}^{\bf PP}$ and  ${\bf B_{|0|<|1|}P}\subseteq{\bf P}^{\bf PP}$. Because a ${\bf P}^{\bf PP}$ machine can recover the exact value of a ${\bf \#P}$ function using threshold queries and then inspect its binary representation directly. We combine these results and show that the comparison based bit-counting complexity classes are Turing equivalent to ${\bf PP}$, that is ${\bf P}^{\bf PP} = {\bf P}^{{\bf B_{|0|=|1|}P}}={\bf P}^{{\bf B_{|0|>|1|}P}}={\bf P}^{{\bf B_{|0|<|1|}P}}$. In essence, their polynomial time Turing closures coincide.

We finally consider parity based bit-counting complexity classes that check the parity of the count of 0's bits and the parity of the count of 1's bits as follows: 8)${\bf B_{|0|\oplus}P}$ and 9)${\bf B_{|1|\oplus}P}$. We establish the bounds ${\bf NP}\cup{\bf coNP}\subseteq{\bf B_{|0|\oplus}P}$ and  ${\bf NP}\cup{\bf coNP}\subseteq{\bf B_{|1|\oplus}P}$. We also show that the parity based bit-counting complexity classes can be recognized by a polynomial time Turing machine with access to oracle ${\bf PP}$, that is  ${\bf B_{|0|\oplus}P}\subseteq{\bf P}^{\bf PP}$ and ${\bf B_{|1|\oplus}P}\subseteq{\bf P}^{\bf PP}$. We also show that when one of the parity based bit-counting complexity class is provided as an oracle for a polynomial time Turing machine, then it can simulate the other one, that is  ${\bf B_{|1| \oplus}P}\subseteq {\bf P}^{{\bf B_{|0| \oplus}P}}$ and ${\bf B_{|0| \oplus}P}\subseteq {\bf P}^{{\bf B_{|1| \oplus}P}}$. These proofs subsequently yield ${\bf P}^{{\bf B_{|0|\oplus}P}}={\bf P}^{{\bf B_{|1|\oplus}P}}\subseteq {\bf P}^{\bf PP}$. 

These results place our family of bit-counting complexity classes within the landscape of complexity hierarchy and show that predicates applied to the binary representation of the accepting path count can encode several classical non-deterministic and probabilistic counting phenomena.

\section{Some classical complexity classes}

\begin{definition}\normalfont
A language $L$ is in complexity class {\bf NP}, if there exists a polynomial $p$ and a polynomial time predicate $R$ such that, for each $x$,

$x \in L \Leftrightarrow ||\{y| \ |y| = p(|x|) \wedge R(x, y)\}|| >0 $
\end{definition}

\begin{definition}\normalfont
A language $L$ is in complexity class {\bf CoNP}, if there exists a polynomial $p$ and a polynomial time predicate $R$ such that, for each $x$,

$x \in L \Leftrightarrow ||\{y| \ |y| = p(|x|) \wedge R(x, y)\}|| =0 $
\end{definition}

\begin{definition}\normalfont
A language $L$ is in complexity class {\bf US}, as defined in \cite{BG82}, if there exists a polynomial $p$ and a polynomial time predicate $R$ such that, for each $x$,

$x \in L \Leftrightarrow ||\{y| \ |y| = p(|x|) \wedge R(x, y)\}|| = 1 $
\end{definition}

\begin{definition}\normalfont
A language $L$ is in complexity class {\bf $\oplus$P}, as defined in \cite{PZ83}, if there exists a polynomial $p$ and a polynomial time predicate $R$ such that, for each $x$,

$x \in L \Leftrightarrow ||\{y| \ |y| = p(|x|) \wedge R(x, y)\}|| \not \equiv$ {\rm 0 (Mod 2)}
\end{definition}

\begin{definition}\normalfont
A language $L$ is in complexity class {\rm \bf C$_{=}$P}, as defined in \cite{S75}, if there exists a polynomial $p$ and a polynomial time predicate $R$ such that, for each $x$,
$x \in L \Leftrightarrow ||\{y| \ |y| = p(|x|) \wedge R(x, y)\}|| = 2^{p(|x|)-1} $
\end{definition}

\begin{definition}\normalfont
A language $L$ is in complexity class {\rm \bf ES}, as defined in \cite{BHR00}, if there exists a polynomial $p$ and a polynomial time predicate $R$ such that, for each $x$,

$x \in L \Leftrightarrow ||\{y| \ |y| = p(|x|) \wedge R(x, y)\}|| = 2^{t}$, where t $\in \mathbb{N}_{0} = \{0,1,2, ...\}$
\end{definition}

\begin{definition}\normalfont A language $L$ is in complexity class {\bf MNS}, as defined in \cite{CP18}, if there exists a polynomial $p$ and a polynomial time predicate $R$ such that, for each $x$,

$x \in L \Leftrightarrow ||\{y| \ |y| = p(|x|) \wedge R(x, y)\}|| = 2^{t}-1$, where t $\in \mathbb{N}_{>0} = \{1,2,3, ...\}$
\end{definition}

\begin{definition}\normalfont
A language $L$ is in complexity class {\rm \bf PP}, as defined in \cite{S75}, if there exists a polynomial $p$ and a polynomial time predicate $R$ such that, for each $x$,

$x \in L \Leftrightarrow ||\{y| \ |y| = p(|x|) \wedge R(x, y)\}|| > 2^{p(|x|)-1} $

\end{definition}

\begin{definition}\normalfont
Functional complexity class {\rm \bf \#P}, as defined in \cite{V79}, counts the total number of accepting paths of a non-deterministic polynomial time Turing machine.

{\bf \#P} = $\{f | (\exists$ a non-deterministic polynomial time Turing machine $M) (\forall x)$ $[f (x) = \#accept_{M}(x)]\}$.
\end{definition}

It was shown that ${\bf CoNP} \subseteq {\bf US} \subseteq {\bf MNS} = {\bf ES } = { \bf C_{=}P} \subseteq {\bf PP}$ and that ${\bf P}^{\bf PP} = {\bf P}^{{\bf\#P}}$. Many other complexity classes can be found in \cite{HO02} and \cite{CP18}.

\section{Bit-counting complexity classes}
\subsection{Definitions}
The function $B_{0}$ counts the number of 0's bits and the function $B_{1}$ counts the number of 1's bits in the binary representation of the numbers in the following definitions. Note that $B_{0}(0)=1$. 

\begin{definition}\normalfont
A language $L$ is in complexity class ${\bf B_{|0|=|1|}P}$, if there exist a polynomial $p$ and a polynomial time predicate $R$ such that, for each $x$, 

$x \in L \Leftrightarrow B_{0}(||\{y| \ |y| = p(|x|) \wedge R(x, y)\}||)$ $=$ $B_{1}(||\{y| \ |y| = p(|x|) \wedge R(x, y)\}||)$
\end{definition}

\begin{definition}\normalfont
A language $L$ is in complexity class ${\bf B_{|0|>|1|}P}$, if there exist a polynomial $p$ and a polynomial time predicate $R$ such that, for each $x$, 

$x \in L \Leftrightarrow B_{0}(||\{y| \ |y| = p(|x|) \wedge R(x, y)\}||)$ $>$ $B_{1}(||\{y| \ |y| = p(|x|) \wedge R(x, y)\}||)$
\end{definition}

\begin{definition}\normalfont
A language $L$ is in complexity class ${\bf B_{|0|<|1|}P}$, if there exist a polynomial $p$ and a polynomial time predicate $R$ such that, for each $x$, 

$x \in L \Leftrightarrow B_{0}(||\{y| \ |y| = p(|x|) \wedge R(x, y)\}||)$ $<$ $B_{1}(||\{y| \ |y| = p(|x|) \wedge R(x, y)\}||)$
\end{definition}

\begin{definition}\normalfont
A language $L$ is in complexity class ${\bf B_{|0| \oplus}P}$, if there exist a polynomial $p$ and a polynomial time predicate $R$ such that, for each $x$,  

$x \in L \Leftrightarrow B_{0}(||\{y| \ |y| = p(|x|) \wedge R(x, y)\}||) \not \equiv  {\rm 0 (Mod 2)}$
\end{definition}

\begin{definition}\normalfont
A language $L$ is in complexity class ${\bf B_{|1| \oplus}P}$, if there exist a polynomial $p$ and a polynomial time predicate $R$ such that, for each $x$,  

$x \in L \Leftrightarrow B_{1}(||\{y| \ |y| = p(|x|) \wedge R(x, y)\}||) \not \equiv  {\rm 0 (Mod 2)}$
\end{definition}

\begin{definition}\normalfont
A language $L$ is in complexity class ${\bf B_{|1|=0}P}$, if there exist a polynomial $p$ and a polynomial time predicate $R$ such that, for each $x$,
 
$x \in L \Leftrightarrow B_{1}(||\{y| \ |y| = p(|x|) \wedge R(x, y)\}||)$ $= 0$
\end{definition}

\begin{definition}\normalfont
A language $L$ is in complexity class ${\bf B_{|1|>0}P}$, if there exist a polynomial $p$ and a polynomial time predicate $R$ such that, for each $x$,
 
$x \in L \Leftrightarrow B_{1}(||\{y| \ |y| = p(|x|) \wedge R(x, y)\}||)$ $> 0$
\end{definition}

\begin{definition}\normalfont
A language $L$ is in complexity class ${\bf B_{|0|=0}P}$, if there exist a polynomial $p$ and a polynomial time predicate $R$ such that, for each $x$, 

$x \in L \Leftrightarrow B_{0}(||\{y| \ |y| = p(|x|) \wedge R(x, y)\}||)$ $= 0$
\end{definition}

\begin{definition}\normalfont
A language $L$ is in complexity class ${\bf B_{|0|>0}P}$, if there exist a polynomial $p$ and a polynomial time predicate $R$ such that, for each $x$, 

$x \in L \Leftrightarrow B_{0}(||\{y| \ |y| = p(|x|) \wedge R(x, y)\}||)$ $> 0$
\end{definition}
\subsection{Canonical Complete Problems}

\begin{definition}\normalfont

Given a Boolean formula in conjunctive normal form, does the number of satisfying truth assignments in its binary representation have: 

- no 1's bits? This is ${\bf B_{|1|=0}P}$-SAT.

- at least one 1's bits? This is  ${\bf B_{|1|>0}P}$-SAT.

- no 0's bits? This is ${\bf B_{|0|=0}P}$-SAT.

- at least one 0's bits? This is  ${\bf B_{|0|>0}P}$-SAT.

- equal number of 0's bits and 1's bits? This is ${\bf B_{|0|=|1|}P}$-SAT.

- more 0's bits than 1's bits? This is ${\bf B_{|0|>|1|}P}$-SAT.

- more 1's bits than 0's bits? This is ${\bf B_{|0|<|1|}P}$-SAT.

- odd number of 0's bits? This is ${\bf B_{|0| \oplus}P}$-SAT.

- odd number of 1's bits? This is ${\bf B_{|1| \oplus}P}$-SAT.

\end{definition}

The completeness of all nine problems follows from the same parsimonious reduction, since the reduction preserves the exact count by making the number of satisfying truth assignments equal to the number of accepting computation paths, and thus preserves every predicate considered here on the binary representation of that count.

\begin{theorem}
Each one of the nine problems in the aforementioned definition is complete for its corresponding bit-counting complexity class under polynomial time parsimonious reductions.
\end{theorem}

\begin{proof}
We prove all of the completeness results simultaneously. Let $\mathcal{P}redicate$ be any one of the bit-counting predicates used in the previous definition. We then let ${\bf B} _{\mathcal{P}redicate}{\bf P}$ denote the corresponding bit-counting complexity class and ${\bf B} _{\mathcal{P}redicate}{\bf P}$-SAT denote the corresponding problem.

We first show membership. Given a Boolean formula $F$ in conjunctive normal form, let $\#SAT(F)$ denote the number of satisfying truth assignments of $F$. Since $\#SAT$ is a ${\bf \#P}$ function, then applying the $\mathcal{P}redicate$ to the $\#SAT(F)$ value gives exactly the acceptance condition defining the ${\bf B} _{\mathcal{P}redicate}{\bf P}$-SAT problem. Therefore, ${\bf B} _{\mathcal{P}redicate}{\bf P}$-SAT $\in{\bf B} _{\mathcal{P}redicate}{\bf P}$.

We next show hardness. Let $L\in{\bf B} _{\mathcal{P}redicate}{\bf P}$. Then there exists a ${\bf \#P}$ function $g$ such that, for every input $x$, $x\in L\Longleftrightarrow\mathcal{P}redicate(g(x))$. Since $g\in{\bf \#P}$, there is a non-deterministic polynomial time Turing machine $M$ such that $g(x)=\#acc_M(x)$. By a polynomial time parsimonious reduction, we can construct a Boolean formula $F_x$ in conjunctive normal form such that satisfying truth assignments of $F_x$ are in one to one correspondence with the accepting computation paths of $M$ on input $x$. As a result, $\#SAT(F_x)=\#acc_M(x)=g(x)$. Therefore, $x\in L
\Longleftrightarrow\mathcal{P}redicate(g(x))
\Longleftrightarrow\mathcal{P}redicate(\#SAT(F_x))
\Longleftrightarrow F_x\in{\bf B}_{\mathcal{P}redicate}{\bf P}\text{-SAT}$.

Thus, the mapping $x\mapsto F_x$ is a polynomial time parsimonious reduction from $L$ to ${\bf B}_{\mathcal{P}redicate}{\bf P}$-SAT. Since $L\in{\bf B}_{\mathcal{P}redicate}{\bf P}$ was for any $L$, ${\bf B}_{ \mathcal{P}redicate}{\bf P}$-SAT is ${\bf B}_{ \mathcal{P}redicate}{\bf P}$-hard. Together with membership, ${\bf B}_{ \mathcal{P}redicate}{\bf P}$-SAT is complete for ${\bf B}_{ \mathcal{P}redicate}{\bf P}$. Since this argument does not depend on the particular bit-counting predicate, it applies to all of the problems in the previous definition.$\qed$
\end{proof}

\section{Relationship with classical complexity classes}

\subsection{Equality by definition} 

The following equalities follow from their respective definitions.  

-${\bf B_{|1|=0}P} ={\bf CoNP}$.

-${\bf B_{|1|>0}P} ={\bf NP}$.

-${\bf B_{|0|=0}P} = {\bf MNS}$.
 
-${\bf B_{|0|>0}P} = {\bf CoMNS}$.  

\subsection{Comparison based bit-counting complexity classes}

We establish the location of the comparison based bit-counting complexity classes ${\bf B_{|0|=|1|}P}$, ${\bf B_{|0|>|1|}P}$ and ${\bf B_{|0|<|1|}P}$ within the complexity hierarchy with the following seven theorems. 

\begin{theorem}

${\bf PP}\subseteq {\bf B_{|0|=|1|}P}$

\end{theorem}

\begin{proof}
We use the following definition of ${\bf B_{|0|=|1|}P}$. A language $L$ is in ${\bf B_{|0|=|1|}P}$ if there exists a ${\bf \#P}$ function $H(x)$ such that $x\in L \iff B_0(H(x))=B_1(H(x))$, where $B_0(n)$ and $B_1(n)$ denote the number of 0's and 1's bits, respectively, in the binary representation of $n$.

Let $L\in {\bf PP}$. Then there is a non-deterministic polynomial time Turing machine $M$ and a polynomially bounded function $m(n)$ such that, on every input $x$ of length $n$, the machine $M$ uses exactly $m(n)$ non-deterministic bits. Assume without loss of generality that $m(n)\ge 2$ for every input length $n$. So $M$ has exactly $2^{m(n)}$ computation paths on input $x$ and $x\in L \iff A(x)>2^{m(n)-1}$, where $A(x)=\#\operatorname{acc}_{M}(x)$. We will use $m=m(|x|)$ for readability. Let $R(x)=2^m-A(x)$ be the number of rejecting paths of $M$. The function $R(x)$ is in ${\bf \#P}$, since we can construct a non-deterministic polynomial time Turing machine that simulates $M$ and accepts exactly on the rejecting paths of $M$. We next define $T=2^{m-1}$ and $H(x)=A(x)\cdot T+R(x)+T^2-T$.

We first verify that $H(x)$ is a valid ${\bf \#P}$ function. The function $A(x)\cdot T=A(x)\cdot 2^{m-1}$ is in ${\bf \#P}$, since we can simulate $M$ and, on each accepting path, guess $m-1$ additional non-deterministic bits. The function $R(x)$ is also in ${\bf \#P}$. Finally, the length-dependent term $T^2-T=2^{2m-2}-2^{m-1}$ is a ${\bf \#P}$ function. To see this, construct a non-deterministic polynomial time Turing machine that guesses a string $y\in\{0,1\}^{2m-2}$. It rejects exactly those strings whose first $m-1$ bits are all $0$. There are exactly $2^{m-1}$ such strings, since the remaining $m-1$ bits are free. Therefore, this machine has exactly $2^{2m-2}-2^{m-1}=T^2-T$ accepting paths. Since ${\bf \#P}$ is closed under addition, we have $H(x)\in {\bf \#P}$.

We next prove the main theorem with the following two lemmas. 

\begin{lemma}
If $x\in L$, then $B_0(H(x))=B_1(H(x))$.
\end{lemma}

\begin{proof}
Assume $x\in L$. Since $x\in L \iff A(x)>2^{m-1}$, we have $A(x)>T$. So we may write $A(x)=T+1+t$ for some integer $t$ satisfying $0\le t\le T-1$.
Since the total number of paths is $2^m=2T$, we have $R(x)=2T-A(x)=2T-(T+1+t)=T-1-t$. We now compute $H(x)=A(x)\cdot T+R(x)+T^2-T$. Substituting $A(x)=T+1+t$ and $R(x)=T-1-t$, we get $H(x)=(T+1+t)T+(T-1-t)+T^2-T$. Thus, $H(x)=2T^2+tT+(T-1-t)$.

Since $T=2^{m-1}$, the value $t$ can be written using exactly $m-1$ bits. Write $t=(b_{m-2}b_{m-3}\cdots b_0)_2$. Then $T-1-t$ is the bitwise complement of $t$ using exactly $m-1$ bits. Therefore, the binary representation of $H(x)$ is $10\underbrace{b_{m-2}b_{m-3}\cdots b_0}_{m-1}\underbrace{\overline{b_{m-2}}\overline{b_{m-3}}\cdots \overline{b_0}}_{m-1}$.

The initial block $10$ has one $1$ and one $0$. The final $2m-2$ bits consist of $t$ followed by its bitwise complement, so they contain exactly $m-1$ ones and exactly $m-1$ zeros. As a result, the full binary representation of $H(x)$ has exactly $m$ ones and exactly $m$ zeros. Therefore, $B_0(H(x))=B_1(H(x))$.$\qed$

\end{proof}

\begin{lemma}
If $x\notin L$, then $B_0(H(x))\neq B_1(H(x))$.
\end{lemma}

\begin{proof}
Assume $x\notin L$. Then $A(x)\le 2^{m-1}=T$. 

First suppose $A(x)<T$. Then $H(x)=A(x)T+R(x)+T^2-T$. We get $H(x)=A(x)T+(2T-A(x))+T^2-T$ since $R(x)=2T-A(x)$ and this simplifies to  $H(x)=T^2+T+A(x)(T-1)$. Because $0\le A(x)<T$, we have $T^2+T\le H(x)<2T^2$. Since $T=2^{m-1}$, this means that $2^{2m-2}<H(x)<2^{2m-1}$. As a result, the binary representation of $H(x)$ has exactly $2m-1$ bits. Consequently, it cannot have the same number of 0's and 1's bits since $2m-1$ is odd. As a result, $B_0(H(x))\neq B_1(H(x))$. 

Next suppose $A(x)=T$. Then $R(x)=T$ and $H(x)=T\cdot T+T+T^2-T=2T^2$. Since $T=2^{m-1}$, we have $H(x)=2^{2m-1}$. The binary representation of $2^{2m-1}$ is a single $1$ followed by $2m-1$ zeros, which means that $B_1(H(x))=1$ and $B_0(H(x))=2m-1$. We have $2m-1\neq 1$ since we already assumed that $m\ge 2$. As a result, $B_0(H(x))\neq B_1(H(x))$. Therefore, if $x\notin L$, then $B_0(H(x))\neq B_1(H(x))$.$\qed$

\end{proof}

The two lemmas show that $x\in L \iff B_0(H(x))=B_1(H(x))$. Since $H(x)\in {\bf \#P}$, we have that $L\in {\bf B_{|0|=|1|}P}$. Since $L\in{\bf PP}$ was any $L$, then we can conclude that ${\bf PP}\subseteq {\bf B_{|0|=|1|}P}$.$\qed$

\end{proof}

\begin{theorem}

${\bf PP}\subseteq {\bf B_{|0|>|1|}P}$

\end{theorem}

\begin{proof}
We use the following definition of ${\bf B_{|0|>|1|}P}$. A language $L$ is in ${\bf B_{|0|>|1|}P}$ if there exists a ${\bf \#P}$ function $H(x)$ such that $x\in L \iff B_0(H(x))>B_1(H(x))$, where $B_0(n)$ and $B_1(n)$ denote the number of 0's and 1's bits, respectively, in the binary representation of $n$.

Let $L\in {\bf PP}$. Then there is a non-deterministic polynomial time Turing machine $M$ and a polynomially bounded function $m(n)$ such that, on every input $x$ of length $n$, the machine $M$ uses exactly $m(n)$ non-deterministic bits. Assume without loss of generality that $m(n)\ge 2$ for every input length $n$. So $M$ has exactly $2^{m(n)}$ computation paths on input $x$ and $x\in L \iff A(x)>2^{m(n)-1}$, where $A(x)=\#\operatorname{acc}_{M}(x)$. We will use $m=m(|x|)$ for readability. We define $T=2^{m-1}$, $\ell=2m+1$ and $H(x)=A(x)+2^\ell-T-1$.

We first verify that $H(x)$ is a valid $\#P$ function. The function $A(x)$ is in ${\bf \#P}$ by definition. Also, the length-dependent term $2^\ell-T-1=2^\ell-2^{m-1}-1$ is a ${\bf \#P}$ function. To see this, construct a non-deterministic polynomial time Turing machine that guesses a string $y\in\{0,1\}^{\ell}$. It rejects exactly those strings whose first $\ell-(m-1)=\ell-m+1$ bits are all $0$ and it also rejects one additional string, say $10\cdots 0$. These two rejection conditions are disjoint. The first condition rejects exactly $2^{m-1}$ strings, since the remaining $m-1$ bits are free and the second condition rejects exactly one additional string. Therefore, this machine has exactly $2^\ell-2^{m-1}-1$ accepting paths. Since ${\bf \#P}$ is closed under addition, we have $H(x)=A(x)+2^\ell-2^{m-1}-1\in {\bf \#P}$.

We now prove the main theorem with the following two lemmas.

\begin{lemma}
If $x\in L$, then $B_0(H(x))>B_1(H(x))$.
\end{lemma}

\begin{proof}
Assume $x\in L$. Since $x\in L \iff A(x)>2^{m-1}$, we have $A(x)>T$. So we may write $A(x)=T+s$ for some integer $s\ge 1$. Since $A(x)\le 2^m$, we also have $s\le 2^m-T=T=2^{m-1}$. Now $H(x)=A(x)+2^\ell-T-1=(T+s)+2^\ell-T-1=2^\ell+(s-1)$. Let $t=s-1$. Then $0\le t\le 2^{m-1}-1$. Writing $t$ in binary using exactly $m-1$ bits as $t=(b_{m-2}b_{m-3}\cdots b_0)_2$, the binary representation of $H(x)=2^\ell+t$ is $1\underbrace{00\cdots 0}_{\ell-m+1}\underbrace{b_{m-2}b_{m-3}\cdots b_0}_{m-1}$. 

Thus, the binary representation of $H(x)$ has a leading $1$ in position $\ell$, followed by $\ell-m+1$ forced zeros, followed by the $m-1$ lower-order bits encoding $t$. The lower-order part contributes at most $m-1$ ones. When we also include the leading $1$, we get at most $1+(m-1)=m$ ones. Thus, $B_1(H(x))\le 1+(m-1)=m$. Also, the forced middle block contributes exactly $\ell-m+1$ zeros. So then we get at least $\ell-m+1=(2m+1)-m+1=m+2$ zeros, since $\ell=2m+1$. Thus, $B_0(H(x))\ge m+2$. As a result, we get $B_1(H(x))\le m<m+2\le B_0(H(x))$. Therefore, $B_0(H(x))>B_1(H(x))$.$\qed$

\end{proof}

\begin{lemma}
If $x\notin L$, then $B_0(H(x)) < B_1(H(x))$.
\end{lemma}

\begin{proof}
Assume $x\notin L$. Then $A(x)\le 2^{m-1}=T$. Therefore, $0\le A(x)\le T$. We have $2^\ell-T-1\le H(x)\le 2^\ell-1$ since $H(x)=2^\ell-T-1+A(x)$. This becomes $2^\ell-2^{m-1}-1\le H(x)\le 2^\ell-1$, since $T=2^{m-1}$. Note that the lower endpoint $2^\ell-2^{m-1}-1$ has binary representation $\underbrace{11\cdots 1}_{\ell-m}0\underbrace{11\cdots 1}_{m-1}$, while the upper endpoint $2^\ell-1$ has binary representation $\underbrace{11\cdots 1}_{\ell}$. As a result, every number in the interval $[2^\ell-2^{m-1}-1,\,2^\ell-1]$ has binary length exactly $\ell$ and begins with at least $\ell-m$ consecutive $1$'s. The remaining $m$ lower-order bits may vary, and they may contain several $0$'s bits, one $0$ bit, or no $0$'s bits. Thus, $B_1(H(x))\ge \ell-m$. We have $\ell-m=(2m+1)-m=m+1$, since $\ell=2m+1$. So then we have $B_1(H(x))\ge m+1$. Also, since the binary representation of $H(x)$ has exactly $\ell=2m+1$ bits, the number of zero bits is at most $(2m+1)-(m+1)=m$. Thus, $B_0(H(x))\le m$. As a result, we get $B_0(H(x))\le m<m+1\le B_1(H(x))$. Therefore, $B_0(H(x))< B_1(H(x))$.$\qed$
\end{proof}

The two lemmas show that $x\in L \iff B_0(H(x))>B_1(H(x))$. Since $H(x)\in {\bf \#P}$, we have that $L\in {\bf B_{|0|>|1|}P}$. Since $L\in{\bf PP}$ was any $L$, then we can conclude that ${\bf PP}\subseteq {\bf B_{|0|>|1|}P}$.$\qed$

\end{proof}

\begin{theorem}

${\bf PP}\subseteq {\bf B_{|0|<|1|}P}$

\end{theorem}

\begin{proof}
We use the following definition of ${\bf B_{|0|<|1|}P}$. A language $L$ is in ${\bf B_{|0|<|1|}P}$ if there exists a ${\bf \#P}$ function $H(x)$ such that $x\in L \iff B_1(H(x))>B_0(H(x))$, where $B_1(n)$ and $B_0(n)$ denote the number of 1's and 0's bits, respectively, in the binary representation of $n$.

Let $L\in {\bf PP}$. Then there is a non-deterministic polynomial time Turing machine $M$ and a polynomially bounded function $m(n)$ such that, on every input $x$ of length $n$, the machine $M$ uses exactly $m(n)$ non-deterministic bits. Assume without loss of generality that $m(n)\ge 2$ for every input length $n$. So $M$ has exactly $2^{m(n)}$ computation paths on input $x$ and $x\in L \iff A(x)>2^{m(n)-1}$, where $A(x)=\#\operatorname{acc}_{M}(x)$. We will use $m=m(|x|)$ for readability. Let $R(x)=2^m-A(x)$ be the number of rejecting paths of $M$. The function $R(x)$ is in ${\bf \#P}$, since we can construct a non-deterministic Turing machine that simulates $M$ and accepts exactly on the rejecting paths of $M$. We next define $T=2^{m-1}$, $\ell=2m+1$ and $H(x) = R(x)+2^\ell-T$.

We first verify that $H(x)$ is a valid $\#P$ function. The function $R(x)$ is in ${\bf \#P}$. Also, the length dependent term $2^\ell-T = 2^\ell-2^{m-1}$ is a ${\bf \#P}$ function. To see this, construct a non-deterministic polynomial time Turing machine that guesses a string $y\in\{0,1\}^{\ell}$. It rejects exactly those strings whose first $\ell-(m-1)=\ell-m+1$ bits are all $0$. There are exactly $2^{m-1}$ such strings, since the remaining $m-1$ bits are free. Therefore, this machine has exactly $2^\ell-2^{m-1}$ accepting paths. Since ${\bf \#P}$ is closed under addition, we have $H(x)=R(x)+2^\ell-2^{m-1}\in {\bf \#P}$.

We now prove the main theorem with the following two lemmas. 

\begin{lemma}
If $x\in L$, then $B_1(H(x))>B_0(H(x))$.
\end{lemma}

\begin{proof}
Assume $x\in L$. Since $x\in L \iff A(x)>2^{m-1}$, we have $A(x)>T$. Therefore, $R(x)=2^m-A(x)<2^m-T=T$. So $0\le R(x)\le T-1$. Thus, $H(x)=2^\ell-T+R(x)$ lies in the interval $2^\ell-T \le H(x) \le 2^\ell-1$. Since $T=2^{m-1}$, the lower endpoint is $2^\ell-2^{m-1}$. Its binary representation is $\underbrace{11\cdots 1}_{\ell-m+1} \underbrace{00\cdots 0}_{m-1}$. 

Every number in the interval $[2^\ell-2^{m-1},\,2^\ell-1]$ has at least the first $\ell-m+1$ bits equal to $1$. Thus, $B_1(H(x))\ge \ell-m+1$. We get $\ell-m+1=(2m+1)-m+1=m+2$, since $\ell=2m+1$. Thus, $B_1(H(x))\ge m+2$. All possible zeros occur only among the final $m-1$ lower-order positions. Thus, $B_0(H(x))\le m-1$. As a result, we get $B_0(H(x))\le m-1 < m+2 \le B_1(H(x))$. Therefore, $B_1(H(x))>B_0(H(x))$.$\qed$

\end{proof}

\begin{lemma}
If $x\notin L$, then $B_1(H(x)) < B_0(H(x))$.
\end{lemma}

\begin{proof}
Assume $x\notin L$. Then $A(x)\le 2^{m-1}=T$. Therefore, $R(x)=2^m-A(x)\ge 2^m-T=T$. So we may write $R(x)=T+s$ for some integer $s\ge 0$. Since $R(x)\le 2^m$, we also have $s\le 2^m-T=T=2^{m-1}$. Now $H(x) = R(x)+2^\ell-T = (T+s)+2^\ell-T = 2^\ell+s$. Writing $s$ in binary using exactly $m$ bits as $s=(b_{m-1}b_{m-2}\cdots b_0)_2$, the binary representation of $H(x)=2^\ell+s$ is $1\underbrace{00\cdots 0}_{\ell-m}\underbrace{b_{m-1}b_{m-2}\cdots b_0}_{m}$. 

Thus, the binary representation of $H(x)$ has a leading $1$ in position $\ell$, followed by $\ell-m$ forced zeros, followed by the $m$ lower-order bits encoding $s$. Since $0\le s\le 2^{m-1}$, the lower-order part contributes at most $m-1$ ones. As a result, $B_1(H(x))\le 1+(m-1)=m$. The binary length of $H(x)$ is $\ell+1$. Thus, $B_0(H(x)) = (\ell+1)-B_1(H(x))$. Using $B_1(H(x))\le m$, we obtain $B_0(H(x)) \ge (\ell+1)-m$. We get $B_0(H(x)) \ge (2m+2)-m = m+2$ since $\ell=2m+1$. As a result, we get $B_1(H(x))\le m < m+2 \le B_0(H(x)) $. Therefore, $B_1(H(x))<B_0(H(x))$. $\qed$

\end{proof}

The two lemmas show that $x\in L \iff B_1(H(x))>B_0(H(x))$. Since $H(x)\in {\bf \#P}$, we have that $L\in {\bf B_{|0|<|1|}P}$. Since $L\in{\bf PP}$ was any $L$, then we can conclude that ${\bf PP}\subseteq {\bf B_{|0|<|1|}P}$. $\qed$

\end{proof}

\begin{theorem}

${\bf B_{|0|=|1|}P}\subseteq {\bf P}^{\bf PP}$

\end{theorem}

\begin{proof}
Let $L\in{\bf B_{|0|=|1|}P}$. Then there exists a ${\bf \#P}$ function $H(x)$ such that $x\in L \iff B_0(H(x))=B_1(H(x))$, where $B_0(n)$ and $B_1(n)$ denote the number of $0$'s and $1$'s bits, respectively, in the standard binary representation of $n$. Since $H(x)$ is a ${\bf \#P}$ function, there is a non-deterministic polynomial time Turing machine $M$ and a polynomial $q$ such that, without loss of generality, on every input $x$ of length $n$, the machine $M$ has exactly $2^{q(n)}$ computation paths and $H(x)=\#\operatorname{acc}_{M}(x)$. Thus, $0\le H(x)\le 2^{q(|x|)}$. We will use $q=q(|x|)$ for readability.

We claim that a ${\bf P}^{\bf PP}$ machine can recover the exact value of $H(x)$. We can observe this by considering threshold questions of the form $H(x)\ge z$, where $1\le z\le 2^q$. Each such threshold question is in ${\bf PP}$. To see this, construct an auxiliary non-deterministic machine with exactly $2^{q+1}$ computation paths as follows: One half of the paths simulate $M(x)$ and contribute exactly $H(x)$ accepting paths, while the other half accepts exactly $2^q-z+1$ paths. Therefore, the auxiliary machine has exactly $H(x)+(2^q-z+1)$ accepting paths. Its majority threshold is $2^q$, and $H(x)+(2^q-z+1)>2^q$ if and only if $H(x)\ge z$. Thus, the threshold language $\{(x,z):H(x)\ge z\}$ is in ${\bf PP}$.

Using this ${\bf PP}$ oracle, a deterministic polynomial time machine can recover $H(x)$ by binary search over the interval $0\le H(x)\le 2^{q}$. Once it has recovered $H(x)$, it can write down the binary representation of $H(x)$, compute $B_0(H(x))$ and $B_1(H(x))$ and accept exactly when $B_0(H(x))=B_1(H(x))$. Therefore, $L\in{\bf P}^{\bf PP}$. Since $L\in{\bf B_{|0|=|1|}P}$ was any $L$, then we can conclude that ${\bf B_{|0|=|1|}P}\subseteq{\bf P}^{\bf PP}$. $\qed$

\end{proof}

\begin{theorem}
${\bf B_{|0|>|1|}P}\subseteq {\bf P}^{\bf PP}$
\end{theorem}
\begin{proof}
Let $L\in{\bf B_{|0|>|1|}P}$. Then there exists a ${\bf \#P}$ function $H(x)$ such that $x\in L \iff B_0(H(x))>B_1(H(x))$, where $B_0(n)$ and $B_1(n)$ denote the number of $0$'s and $1$'s bits, respectively, in the standard binary representation of $n$. Since $H(x)$ is a ${\bf \#P}$ function, there is a non-deterministic polynomial time Turing machine $M$ and a polynomial $q$ such that, without loss of generality, on every input $x$ of length $n$, the machine $M$ has exactly $2^{q(n)}$ computation paths and $H(x)=\#\operatorname{acc}_{M}(x)$. Thus, $0\le H(x)\le 2^{q(|x|)}$. We will use $q=q(|x|)$ for readability.

As shown in the previous theorem, threshold questions of the form $H(x)\ge z$, where $1\le z\le 2^q$, are in ${\bf PP}$. Therefore, a deterministic polynomial time machine with access to a ${\bf PP}$ oracle can recover the exact value of $H(x)$ by binary search over the interval $0\le H(x)\le 2^q$. After recovering $H(x)$, the machine computes its standard binary representation, counts the number of $0$'s and $1$'s bits, and accepts exactly when $B_0(H(x))>B_1(H(x))$. Therefore, $L\in{\bf P}^{\bf PP}$. Since $L\in{\bf B_{|0|>|1|}P}$ was any $L$, then we can conclude that ${\bf B_{|0|>|1|}P}\subseteq{\bf P}^{\bf PP}$.$\qed$
\end{proof}

\begin{theorem}

${\bf B_{|0|<|1|}P}\subseteq {\bf P}^{\bf PP}$

\end{theorem}

\begin{proof}
Let $L\in{\bf B_{|0|<|1|}P}$. Then there exists a ${\bf \#P}$ function $H(x)$ such that $x\in L \iff B_0(H(x))<B_1(H(x))$, where $B_0(n)$ and $B_1(n)$ denote the number of $0$'s and $1$'s bits, respectively, in the standard binary representation of $n$. Since $H(x)$ is a ${\bf \#P}$ function, there is a non-deterministic polynomial time Turing machine $M$ and a polynomial $q$ such that, without loss of generality, on every input $x$ of length $n$, the machine $M$ has exactly $2^{q(n)}$ computation paths and $H(x)=\#\operatorname{acc}_{M}(x)$. Thus, $0\le H(x)\le 2^{q(|x|)}$. We will use $q=q(|x|)$ for readability.

As shown in previous theorems, threshold questions of the form $H(x)\ge z$, where $1\le z\le 2^q$, are in ${\bf PP}$. Therefore, a deterministic polynomial time machine with access to a ${\bf PP}$ oracle can recover the exact value of $H(x)$ by binary search over the interval $0\le H(x)\le 2^q$. After recovering $H(x)$, the machine computes its standard binary representation, counts the number of 1's and 0's bits and accepts exactly when $B_0(H(x))<B_1(H(x))$. Therefore, $L\in{\bf P}^{\bf PP}$. Since $L\in{\bf B_{|0|<|1|}P}$ was any $L$, then we can conclude that ${\bf B_{|0|<|1|}P}\subseteq{\bf P}^{\bf PP}$.$\qed$

\end{proof}

\begin{theorem}

${\bf P}^{{\bf B_{|0|=|1|}P}}={\bf P}^{{\bf B_{|0|>|1|}P}}={\bf P}^{{\bf B_{|0|<|1|}P}}={\bf P}^{\bf PP}$

\end{theorem}

\begin{proof}

Let $\mathcal{C}$ be any one of the following three classes:

$\mathcal{C}\in\{{\bf B_{|0|=|1|}P},{\bf B_{|0|>|1|}P},{\bf B_{|0|<|1|}P}\}$. 

From the previous containment theorems, we have $\mathcal{C}\subseteq{\bf P}^{\bf PP}$. Also, from the earlier theorems, we have ${\bf PP}\subseteq\mathcal{C}$. We now show that ${\bf P}^{\mathcal{C}}={\bf P}^{\bf PP}$. 

First, since ${\bf PP}\subseteq\mathcal{C}$, any ${\bf P}^{\bf PP}$ computation can be simulated by a ${\bf P}^{\mathcal{C}}$ computation. Therefore, ${\bf P}^{\bf PP}\subseteq{\bf P}^{\mathcal{C}}$. Second, since $\mathcal{C}\subseteq{\bf P}^{\bf PP}$, any oracle query to a language in $\mathcal{C}$ can be simulated by a deterministic polynomial time computation with access to a ${\bf PP}$ oracle. Therefore, ${\bf P}^{\mathcal{C}}\subseteq{\bf P}^{{\bf P}^{\bf PP}}$. However, nested deterministic polynomial time oracle computations collapse, so ${\bf P}^{{\bf P}^{\bf PP}}={\bf P}^{\bf PP}$. Thus, ${\bf P}^{\mathcal{C}}\subseteq{\bf P}^{\bf PP}$. Combining the two containments gives ${\bf P}^{\mathcal{C}}={\bf P}^{\bf PP}$. 

Since $\mathcal{C}$ was any one of $\{{\bf B_{|0|=|1|}P},{\bf B_{|0|>|1|}P},{\bf B_{|0|<|1|}P}\}$, then we can conclude that ${\bf P}^{{\bf B_{|0|=|1|}P}}={\bf P}^{{\bf B_{|0|>|1|}P}}={\bf P}^{{\bf B_{|0|<|1|}P}}={\bf P}^{\bf PP}$.$\qed$

\end{proof}

\subsection{Parity based bit-counting complexity classes}

We establish the location of the parity based bit-counting complexity classes ${\bf B_{|0| \oplus}P}$ and ${\bf B_{|1| \oplus}P}$ within the complexity hierarchy with the following ten theorems. 

\begin{theorem}
${\bf B_{|1|>0}P} = {\bf NP}\subseteq {\bf B_{|1| \oplus}P}$
\end{theorem}
\begin{proof}
We use the following definition of ${\bf B_{|1| \oplus}P}$. A language $L$ is in ${\bf B_{|1| \oplus}P}$ if there exists a ${\bf \#P}$ function $H(x)$ such that $x\in L \iff B_1(H(x))\not\equiv 0 \pmod 2$, where $B_1(n)$ denotes the number of $1$'s bits in the binary representation of $n$.
Let $L\in{\bf NP}$. Then there is a polynomial time predicate $R$ and a polynomial $p$ such that, for every input $x$, $x\in L \iff \exists y \text{ such that } |y|=p(|x|) \text{ and } R(x,y)$. Define $f(x)=||\{y \mid |y|=p(|x|) \wedge R(x,y)\}||$. Then $f(x)$ is a ${\bf \#P}$ function and $x\in L \iff f(x)>0$. We will write $p=p(|x|)$ for readability. We define $r=2p+1$, which is odd. We also have $f(x)<2^r$, since $f(x)\le 2^p$ and $p<r$. We next define $H(x)=(2^r-1)f(x)$. 

We first verify that $H(x)$ is a valid ${\bf \#P}$ function. Construct a non-deterministic polynomial time Turing machine as follows. On input $x$, it first guesses a string $y\in\{0,1\}^{p}$. If $R(x,y)$ is false, it rejects. If $R(x,y)$ is true, it guesses another string $z\in\{0,1\}^{r}$ and accepts if and only if $z\neq 00\cdots 0$. For each accepting witness $y$, there are exactly $2^r-1$ accepting choices of $z$. Therefore, the total number of accepting paths is exactly $(2^r-1)f(x)$. So $H(x)\in{\bf \#P}$. 

We next prove the main theorem with the following two lemmas.

\begin{lemma}
If $x\notin L$, then $B_1(H(x))\equiv 0 \pmod 2$.
\end{lemma}
\begin{proof}
Suppose $x\notin L$. Then $f(x)=0$ and $H(x)=0$. The binary representation of $0$ is $0$, so $B_1(H(x))=0$. Therefore, $B_1(H(x))\equiv 0 \pmod 2$.$\qed$
\end{proof}

\begin{lemma}
If $x\in L$, then $B_1(H(x))\not\equiv 0 \pmod 2$.
\end{lemma}
\begin{proof}
Suppose $x\in L$. Then $f(x)>0$, so $1\le f(x)<2^r$. We can rewrite $H(x)$ as $H(x)=f(x)2^r-f(x)=(f(x)-1)2^r+(2^r-f(x))$. We can then write $f(x)-1$ as an $r$-bit string $a_{r-1}a_{r-2}\cdots a_0$, because $0\le f(x)-1<2^r$. Also, $2^r-f(x)=2^r-1-(f(x)-1)$, so the $r$-bit string for $2^r-f(x)$ is the bitwise complement $\overline{a}_{r-1}\overline{a}_{r-2}\cdots\overline{a}_0$, where $\overline{a}_i=1-a_i$. Therefore, $H(x)$ is represented by the two consecutive $r$-bit blocks $\underbrace{a_{r-1}a_{r-2}\cdots a_0}_{f(x)-1}\underbrace{\overline{a}_{r-1}\overline{a}_{r-2}\cdots\overline{a}_0}_{2^r-f(x)}$. 

Note that the term $(f(x)-1)2^r$ shifts the first block $r$ positions to the left, and the term $2^r-f(x)$ occupies the lower $r$ positions. For each $i$, exactly one of $a_i$ and $\overline{a}_i$ is equal to $1$. So the two blocks contain exactly $r$ many $1$'s altogether. Thus $B_1(H(x))=r$, which is odd since $r=2p+1$. Therefore, $B_1(H(x))\not\equiv 0 \pmod 2$.$\qed$
\end{proof}

By the two lemmas, $x\in L \iff B_1(H(x))\not\equiv 0 \pmod 2$. Therefore, $L\in{\bf B_{|1| \oplus}P}$. Since $L\in{\bf NP}$ was any $L$, then we can conclude that ${\bf NP}\subseteq {\bf B_{|1| \oplus}P}$.$\qed$
\end{proof}

\begin{theorem}
${\bf B_{|1|>0}P} ={\bf NP}\subseteq {\bf B_{|0| \oplus}P}$
\end{theorem}
\begin{proof}
We use the following definition of ${\bf B_{|0| \oplus}P}$. A language $L$ is in ${\bf B_{|0| \oplus}P}$ if there exists a ${\bf \#P}$ function $H(x)$ such that $x\in L \iff B_0(H(x))\not\equiv 0 \pmod 2$, where $B_0(n)$ denotes the number of $0$'s bits in the binary representation of $n$.
Let $L\in{\bf NP}$. Then there is a polynomial time predicate $R$ and a polynomial $p$ such that, for every input $x$, $x\in L \iff \exists y \text{ such that } |y|=p(|x|) \text{ and } R(x,y)$. Define $f(x)=||\{y \mid |y|=p(|x|) \wedge R(x,y)\}||$. Then $f(x)$ is a ${\bf \#P}$ function and $x\in L \iff f(x)>0$. We will use $p=p(|x|)$ for readability. We define $r=2p+1$ and $\ell=2r$. Then $r$ is odd and $\ell$ is even. We also have $(2^r-1)f(x)<2^{r+p}<2^\ell$ since $f(x)\le2^p$. We next define $H(x)=2^\ell+(2^r-1)f(x)$.

We first verify that $H(x)$ is a valid ${\bf \#P}$ function. The term $2^\ell$ is a ${\bf \#P}$ function, because a non-deterministic polynomial time Turing machine can guess a string of length $\ell$ and accept on every branch. The term $(2^r-1)f(x)$ is also a ${\bf \#P}$ function such that on input $x$, guess a string $y\in\{0,1\}^{p}$. If $R(x,y)$ is false, reject. If $R(x,y)$ is true, guess a string $z\in\{0,1\}^{r}$ and accept if and only if $z\neq 00\cdots 0$. For each accepting witness $y$, there are exactly $2^r-1$ accepting choices of $z$, so the total number of accepting paths is $(2^r-1)f(x)$. Since ${\bf \#P}$ is closed under addition, we have $H(x)=2^\ell+(2^r-1)f(x)\in{\bf \#P}$.

We next prove the main theorem with the following two lemmas.

\begin{lemma}
If $x\notin L$, then $B_0(H(x))\equiv 0 \pmod 2$.
\end{lemma}
\begin{proof}
Suppose $x\notin L$. Then $f(x)=0$ and $H(x)=2^\ell$. The binary representation of $2^\ell$ is a single $1$ followed by $\ell$ many $0$'s. Thus, $B_0(H(x))=\ell$, which is even since $\ell=2r=4p+2$. Therefore,  $B_0(H(x))\equiv 0 \pmod 2$.$\qed$
\end{proof}

\begin{lemma}
If $x\in L$, then $B_0(H(x))\not\equiv 0 \pmod 2$.
\end{lemma}
\begin{proof}
Suppose $x\in L$. Then $f(x)>0$, so $1\le f(x)<2^r$. Let $K=(2^r-1)f(x)$, so that $H(x)=2^\ell+K$. We first show that $K$ has exactly $r$ many $1$'s in its binary representation. We can rewrite $K$ as $K=(2^r-1)f(x)=f(x)2^r-f(x)=(f(x)-1)2^r+(2^r-f(x))$. We can then write $f(x)-1$ as an $r$-bit string $a_{r-1}a_{r-2}\cdots a_0$, because $0\le f(x)-1<2^r$. Also, $2^r-f(x)=2^r-1-(f(x)-1)$, so the $r$-bit string for $2^r-f(x)$ is the bitwise complement $\overline{a}_{r-1}\overline{a}_{r-2}\cdots\overline{a}_0$, where $\overline{a}_i=1-a_i$. Therefore, $K$ is represented by the two consecutive $r$-bit blocks $\underbrace{a_{r-1}a_{r-2}\cdots a_0}_{f(x)-1}\underbrace{\overline{a}_{r-1}\overline{a}_{r-2}\cdots\overline{a}_0}_{2^r-f(x)}$. 

Note that the term $(f(x)-1)2^r$ shifts the first block $r$ positions to the left, and the term $2^r-f(x)$ occupies the lower $r$ positions. For each $i$, exactly one of $a_i$ and $\overline{a}_i$ is equal to $1$. So the two blocks contain exactly $r$ many $1$'s altogether. Thus $K$ has exactly $r$ many $1$'s in its binary representation. Since $K<2^\ell$, the binary representation of $H(x)=2^\ell+K$ has a leading $1$ in position $\ell$, followed by the actual lower-order bit positions $0,1,\dots,\ell-1$. These lower $\ell$ positions contain exactly the bits of $K$ in their corresponding positions, and every lower position not occupied by a $1$ is a genuine $0$-bit of $H(x)$. Since exactly $r$ of these lower positions are $1$'s, the remaining $\ell-r$ lower positions are $0$'s. Thus $B_0(H(x))=\ell-r=(4p+2)-(2p+1)=2p+1$, which is odd. Therefore, $B_0(H(x))\not\equiv 0 \pmod 2$.$\qed$
\end{proof}

By the two lemmas, $x\in L \iff B_0(H(x))\not\equiv 0 \pmod 2$. Therefore, $L\in{\bf B_{|0| \oplus}P}$. Since $L\in{\bf NP}$ was for any $L$, then we can conclude that ${\bf NP}\subseteq {\bf B_{|0| \oplus}P}$.$\qed$
\end{proof}

\begin{theorem}
${\bf B_{|1|=0}P} ={\bf CoNP}\subseteq {\bf B_{|1|\oplus}P}$
\end{theorem}
\begin{proof}
Let $L\in{\bf CoNP}$. Then there is a polynomial time predicate $R$ and a polynomial $p$ such that, for every input $x$, $x\in L \iff \nexists y \text{ such that } |y|=p(|x|) \text{ and } R(x,y)$. Define $f(x)=||\{y\mid |y|=p(|x|)\wedge R(x,y)\}||$. Then $f(x)$ is a ${\bf \#P}$ function and $x\in L \iff f(x)=0$. We will use $p=p(|x|)$ for readability. We define $r=2p+1$ and $\ell=4p+3$. Then both $r$  and $\ell$ are odd. We have $(2^r-1)f(x)<2^{r+p}<2^\ell$, since $f(x)\le 2^p$. We next define $H(x)=2^\ell+(2^r-1)f(x)$.

We first verify that $H(x)$ is a valid ${\bf \#P}$ function. The term $2^\ell$ is a ${\bf \#P}$ function, since a non-deterministic polynomial time Turing machine can guess a string of length $\ell$ and accept on every branch. The term $(2^r-1)f(x)$ is also a ${\bf \#P}$ function such that for each witness $y$ counted by $f(x)$, guess a string $z\in\{0,1\}^r$ and accept if and only if $z\neq 00\cdots 0$. This gives exactly $2^r-1$ accepting paths for each such $y$. Since ${\bf \#P}$ is closed under addition, $H(x)\in{\bf \#P}$. 

We next prove the main theorem with the following two lemmas.

\begin{lemma}
If $x\in L$, then $B_1(H(x))\not\equiv 0 \pmod 2$.
\end{lemma}
\begin{proof}
Suppose $x\in L$. Then $f(x)=0$ and $H(x)=2^\ell$. The binary representation of $2^\ell$ is a single $1$ followed by $\ell$ many $0$'s. Thus $B_1(H(x))=1$, which is odd. Therefore, $B_1(H(x))\not\equiv 0 \pmod 2$.$\qed$
\end{proof}

\begin{lemma}
If $x\notin L$, then $B_1(H(x))\equiv 0 \pmod 2$.
\end{lemma}
\begin{proof}
Suppose $x\notin L$. Then $f(x)>0$, so $1\le f(x)<2^r$. Let $K=(2^r-1)f(x)$, so that $H(x)=2^\ell+K$. We first show that $K$ has exactly $r$ many $1$'s in its binary representation. We can rewrite $K$ as $K=(2^r-1)f(x)=f(x)2^r-f(x)=(f(x)-1)2^r+(2^r-f(x))$. We can then write $f(x)-1$ as an $r$-bit string $a_{r-1}a_{r-2}\cdots a_0$, because $0\le f(x)-1<2^r$. Also, $2^r-f(x)=2^r-1-(f(x)-1)$, so the $r$-bit string for $2^r-f(x)$ is the bitwise complement $\overline{a}_{r-1}\overline{a}_{r-2}\cdots\overline{a}_0$, where $\overline{a}_i=1-a_i$. Therefore, $K$ is represented by the two consecutive $r$-bit blocks $\underbrace{a_{r-1}a_{r-2}\cdots a_0}_{f(x)-1}\underbrace{\overline{a}_{r-1}\overline{a}_{r-2}\cdots\overline{a}_0}_{2^r-f(x)}$.

Note that the term $(f(x)-1)2^r$ shifts the first block $r$ positions to the left, and the term $2^r-f(x)$ occupies the lower $r$ positions. For each $i$, exactly one of $a_i$ and $\overline{a}_i$ is equal to $1$. So the two blocks contain exactly $r$ many $1$'s altogether. Thus $K$ has exactly $r$ many $1$'s in its binary representation. Since $K<2^\ell$, adding $2^\ell$ creates one new leading $1$ and does not interfere with the lower-order bits. Thus, $B_1(H(x))=1+r$ is even, since $r=2p+1$ is odd. Therefore, $B_1(H(x))\equiv 0 \pmod 2$.$\qed$
\end{proof}

By the two lemmas, $x\in L \iff B_1(H(x))\not\equiv 0 \pmod 2$. Therefore, $L\in{\bf B_{|1|\oplus}P}$. Since $L\in{\bf CoNP}$ was any $L$, then we can conclude that ${\bf CoNP}\subseteq {\bf B_{|1|\oplus}P}$.$\qed$
\end{proof}

\begin{theorem}
${\bf B_{|1|=0}P} ={\bf CoNP}\subseteq {\bf B_{|0|\oplus}P}$
\end{theorem}
\begin{proof}
Let $L\in{\bf CoNP}$. Then there is a polynomial time predicate $R$ and a polynomial $p$ such that, for every input $x$, $x\in L \iff \nexists y \text{ such that } |y|=p(|x|) \text{ and } R(x,y)$. Define $f(x)=||\{y\mid |y|=p(|x|)\wedge R(x,y)\}||$. Then $f(x)$ is a ${\bf \#P}$ function and $x\in L \iff f(x)=0$. We will use $p=p(|x|)$ for readability. We define $r=2p+1$ and $\ell=4p+3$. Then both $r$ and $\ell$ are odd. We have $(2^r-1)f(x)<2^{r+p}<2^\ell$, since $f(x)\le 2^p$. We next define $H(x)=2^\ell+(2^r-1)f(x)$.

We first verify that $H(x)$ is a valid ${\bf \#P}$ function. The term $2^\ell$ is a ${\bf \#P}$ function, since a non-deterministic polynomial time Turing machine can guess a string of length $\ell$ and accept on every branch. The term $(2^r-1)f(x)$ is also a ${\bf \#P}$ function such that for each witness $y$ counted by $f(x)$, guess a string $z\in\{0,1\}^r$ and accept if and only if $z\neq 00\cdots 0$. This gives exactly $2^r-1$ accepting paths for each such $y$. Since ${\bf \#P}$ is closed under addition, $H(x)\in{\bf \#P}$. 

We next prove the main theorem with the following two lemmas.

\begin{lemma}
If $x\in L$, then $B_0(H(x))\not\equiv 0 \pmod 2$.
\end{lemma}
\begin{proof}
Suppose $x\in L$. Then $f(x)=0$ and $H(x)=2^\ell$. The binary representation of $2^\ell$ is a single $1$ followed by $\ell$ many $0$'s. Thus, $B_0(H(x))=\ell$ which is odd, since $\ell=4p+3$ is odd. Therefore, $B_0(H(x))\not\equiv 0 \pmod 2$.$\qed$
\end{proof}

\begin{lemma}
If $x\notin L$, then $B_0(H(x))\equiv 0 \pmod 2$.
\end{lemma}
\begin{proof}
Suppose $x\notin L$. Then $f(x)>0$, so $1\le f(x)<2^r$. Let $K=(2^r-1)f(x)$, so that $H(x)=2^\ell+K$. We first show that $K$ has exactly $r$ many $1$'s in its binary representation. We can rewrite $K$ as $K=(2^r-1)f(x)=f(x)2^r-f(x)=(f(x)-1)2^r+(2^r-f(x))$. We can then write $f(x)-1$ as an $r$-bit string $a_{r-1}a_{r-2}\cdots a_0$, because $0\le f(x)-1<2^r$. Also, $2^r-f(x)=2^r-1-(f(x)-1)$, so the $r$-bit string for $2^r-f(x)$ is the bitwise complement $\overline{a}_{r-1}\overline{a}_{r-2}\cdots\overline{a}_0$, where $\overline{a}_i=1-a_i$. Therefore, $K$ is represented by the two consecutive $r$-bit blocks $\underbrace{a_{r-1}a_{r-2}\cdots a_0}_{f(x)-1}\underbrace{\overline{a}_{r-1}\overline{a}_{r-2}\cdots\overline{a}_0}_{2^r-f(x)}$. 

Note that the term $(f(x)-1)2^r$ shifts the first block $r$ positions to the left, and the term $2^r-f(x)$ occupies the lower $r$ positions. For each $i$, exactly one of $a_i$ and $\overline{a}_i$ is equal to $1$. So the two blocks contain exactly $r$ many $1$'s altogether. Thus $K$ has exactly $r$ many $1$'s in its binary representation. Since $K<2^\ell$, the binary representation of $H(x)=2^\ell+K$ has one leading $1$ in position $\ell$, followed by the actual lower-order bit positions $0,1,\dots,\ell-1$. These lower $\ell$ positions contain exactly the bits of $K$ in their corresponding positions. Since exactly $r$ of these lower positions are $1$'s, the remaining $\ell-r$ lower positions are $0$'s. Thus $B_0(H(x))=\ell-r=(4p+3)-(2p+1)=2p+2$, which is even. Therefore, $B_0(H(x))\equiv 0 \pmod 2$.$\qed$
\end{proof}

By the two lemmas, $x\in L \iff B_0(H(x))\not\equiv 0 \pmod 2$. Therefore $L\in{\bf B_{|0|\oplus}P}$. Since $L\in{\bf CoNP}$ was any $L$, then we can conclude that ${\bf CoNP}\subseteq {\bf B_{|0|\oplus}P}$.$\qed$
\end{proof}

\begin{theorem}
${\bf B_{|0| \oplus}P}\subseteq {\bf P}^{\bf PP}$
\end{theorem}
\begin{proof}
Let $L\in{\bf B_{|0| \oplus}P}$. Then there exists a ${\bf \#P}$ function $H(x)$ such that $x\in L \iff B_0(H(x))\not\equiv 0 \pmod 2$, where $B_0(n)$ denotes the number of $0$'s bits in the standard binary representation of $n$. Since $H(x)$ is a ${\bf \#P}$ function, there is a non-deterministic polynomial time Turing machine $M$ and a polynomial $q$ such that, without loss of generality, on every input $x$ of length $n$, the machine $M$ has exactly $2^{q(n)}$ computation paths and $H(x)=\#\operatorname{acc}_{M}(x)$. Thus, $0\le H(x)\le 2^{q(|x|)}$. We will use $q=q(|x|)$ for readability.

We claim that a ${\bf P}^{\bf PP}$ machine can recover the exact value of $H(x)$. We can observe this by considering threshold questions of the form $H(x)\ge z$, where $1\le z\le 2^q$. Each such threshold question is in ${\bf PP}$. To see this, construct an auxiliary non-deterministic machine with exactly $2^{q+1}$ computation paths as follows: One half of the paths simulate $M(x)$ and contribute exactly $H(x)$ accepting paths, while the other half accepts exactly $2^q-z+1$ paths. Therefore, the auxiliary machine has exactly $H(x)+(2^q-z+1)$ accepting paths. Its majority threshold is $2^q$, and $H(x)+(2^q-z+1)>2^q$ if and only if $H(x)\ge z$. Thus, the threshold language $\{(x,z):H(x)\ge z\}$ is in ${\bf PP}$.

Using this ${\bf PP}$ oracle, a deterministic polynomial time machine can recover $H(x)$ by binary search over the interval $0\le H(x)\le 2^q$. Once it has recovered $H(x)$, it can write down the standard binary representation of $H(x)$, compute $B_0(H(x))$, and accept exactly when $B_0(H(x))\not\equiv 0 \pmod 2$. Therefore, $L\in{\bf P}^{\bf PP}$. Since $L\in{\bf B_{|0| \oplus}P}$ was any $L$, then we can conclude that ${\bf B_{|0| \oplus}P}\subseteq {\bf P}^{\bf PP}$.$\qed$
\end{proof}

\begin{theorem}
${\bf B_{|1| \oplus}P}\subseteq {\bf P}^{\bf PP}$
\end{theorem}
\begin{proof}
Let $L\in{\bf B_{|1| \oplus}P}$. Then there exists a ${\bf \#P}$ function $H(x)$ such that $x\in L \iff B_1(H(x))\not\equiv 0 \pmod 2$, where $B_1(n)$ denotes the number of $1$'s bits in the standard binary representation of $n$. Since $H(x)$ is a ${\bf \#P}$ function, there is a non-deterministic polynomial time Turing machine $M$ and a polynomial $q$ such that, without loss of generality, on every input $x$ of length $n$, the machine $M$ has exactly $2^{q(n)}$ computation paths and $H(x)=\#\operatorname{acc}_{M}(x)$. Thus, $0\le H(x)\le 2^{q(|x|)}$. We will use $q=q(|x|)$ for readability.

As shown in the previous theorem, threshold questions of the form $H(x)\ge z$, where $1\le z\le 2^q$, are in ${\bf PP}$. Therefore, using a ${\bf PP}$ oracle, a deterministic polynomial time machine can recover the exact value of $H(x)$ by binary search over the interval $0\le H(x)\le 2^q$. Once it has recovered $H(x)$, the machine can write down the standard binary representation of $H(x)$, compute $B_1(H(x))$, and accept exactly when $B_1(H(x))\not\equiv 0 \pmod 2$. Therefore, $L\in{\bf P}^{\bf PP}$. Since $L\in{\bf B_{|1| \oplus}P}$ was any $L$, then we can conclude that ${\bf B_{|1| \oplus}P}\subseteq {\bf P}^{\bf PP}$.$\qed$
\end{proof}
\newpage

\begin{theorem}
${\bf B_{|1| \oplus}P}\subseteq {\bf P}^{{\bf B_{|0| \oplus}P}}$
\end{theorem}
\begin{proof}
Let $L\in{\bf B_{|1| \oplus}P}$. Then there exists a ${\bf \#P}$ function $H(x)$ such that $x\in L \iff B_1(H(x))\not\equiv 0 \pmod 2$, where $B_1(n)$ denotes the number of $1$'s bits in the binary representation of $n$.
Since $H(x)$ is a ${\bf \#P}$ function, there is a non-deterministic polynomial time Turing machine $M$ and a polynomially bounded function $q(n)$ such that, on every input $x$ of length $n$, the machine $M$ uses exactly $q(n)$ non-deterministic bits and $H(x)=\#\operatorname{acc}_{M}(x)$. Thus, $0\le H(x)\le 2^{q(|x|)}$. We will use $q=q(|x|)$ for readability.

We now define a new ${\bf \#P}$ function $G(x)=2^{q+1}+H(x)$.
The function $G(x)$ is in ${\bf \#P}$, because $H(x)$ is in ${\bf \#P}$ and the length-dependent term $2^{q+1}$ is in ${\bf \#P}$. Indeed, a non-deterministic polynomial time Turing machine can guess a string of length $q+1$ and accept on every branch, giving exactly $2^{q+1}$ accepting paths. Since ${\bf \#P}$ is closed under addition, $G(x)\in{\bf \#P}$.

We now consider the following language $O=\{x:B_0(G(x))\not\equiv 0 \pmod 2\}$. By definition, $O\in{\bf B_{|0| \oplus}P}$. We will show that $L$ can be decided in polynomial time with oracle access to $O$.
Since $0\le H(x)\le 2^q$, the value $H(x)$ only affects the lower-order bit positions $0,1,\dots,q$ of $G(x)$. Therefore, the standard binary representation of $G(x)=2^{q+1}+H(x)$ has a leading $1$ in position $q+1$, followed by exactly $q+1$ lower-order bit positions encoding $H(x)$. Among these $q+1$ lower-order positions, exactly $B_1(H(x))$ positions contain $1$'s. Therefore, $B_0(G(x))=(q+1)-B_1(H(x))$. Taking this equation modulo $2$, we get $B_0(G(x))\equiv q+1-B_1(H(x)) \pmod 2$. Since subtraction is the same as addition modulo $2$, this is equivalent to $B_0(G(x))\equiv q+1+B_1(H(x)) \pmod 2$. Therefore, $B_1(H(x))\equiv B_0(G(x))+(q+1) \pmod 2$.

A deterministic polynomial time oracle machine can decide $L$ as follows. On input x, it computes $q = q(|x|)$ and queries the oracle $O$ on input $x$. If $q+1$ is even, then $B_1(H(x))$ and $B_0(G(x))$ have the same parity, so then the machine accepts exactly when the oracle says yes. If $q+1$ is odd, then $B_1(H(x))$ and $B_0(G(x))$ have opposite parity, so then the machine accepts exactly when the oracle says no.

We next prove the main theorem with the following two lemmas

\begin{lemma}
If $x\in L$, then the oracle machine accepts $x$.
\end{lemma}
\begin{proof}
Assume $x\in L$. Since $L\in{\bf B_{|1| \oplus}P}$ via $H(x)$, we have $B_1(H(x))\not\equiv 0 \pmod 2$. Thus, $B_1(H(x))$ is odd. If $q+1$ is even, then $B_1(H(x))$ and $B_0(G(x))$ have the same parity. Thus, $B_0(G(x))$ is odd, so the oracle says yes, and the machine accepts. If $q+1$ is odd, then $B_1(H(x))$ and $B_0(G(x))$ have opposite parity. Thus, $B_0(G(x))$ is even, so the oracle says no, and the machine accepts by construction. Therefore, in both cases, the oracle machine accepts $x$.$\qed$
\end{proof}
\begin{lemma}
If $x\notin L$, then the oracle machine rejects $x$.
\end{lemma}
\begin{proof}
Assume $x\notin L$. Since $L\in{\bf B_{|1| \oplus}P}$ via $H(x)$, we have $B_1(H(x))\equiv 0 \pmod 2$. Thus $B_1(H(x))$ is even. If $q+1$ is even, then $B_1(H(x))$ and $B_0(G(x))$ have the same parity. Thus, $B_0(G(x))$ is even, so the oracle says no, and the machine rejects. If $q+1$ is odd, then $B_1(H(x))$ and $B_0(G(x))$ have opposite parity. Thus, $B_0(G(x))$ is odd, so the oracle says yes, and the machine rejects by construction. Therefore, in both cases, the oracle machine rejects $x$.$\qed$
\end{proof}

The two lemmas show that the deterministic polynomial time oracle machine accepts exactly the strings in $L$. Therefore, $L\in{\bf P}^{{\bf B_{|0| \oplus}P}}$. Since $L\in{\bf B_{|1| \oplus}P}$ was for any $L$, then we can conclude that ${\bf B_{|1| \oplus}P}\subseteq {\bf P}^{{\bf B_{|0| \oplus}P}}$.$\qed$
\end{proof}

\begin{theorem}
${\bf B_{|0|\oplus}P}\subseteq {\bf P}^{{\bf B_{|1|\oplus}P}}$
\end{theorem}
\begin{proof}
Let $L\in{\bf B_{|0|\oplus}P}$. Then there exists a ${\bf \#P}$ function $H(x)$ such that $x\in L \iff B_0(H(x))\not\equiv 0 \pmod 2$. Since $H(x)$ is a ${\bf \#P}$ function, there is a non-deterministic polynomial time Turing machine $M$ and a polynomial $q$ such that, without loss of generality, on every input $x$ of length $n$, the machine $M$ has exactly $2^{q(n)}$ computation paths and $H(x)=\#\operatorname{acc}_{M}(x)$. Thus, $0\le H(x)\le 2^{q(|x|)}$. We will use $q=q(|x|)$ for readability.

Let $\ell(n)$ denote the number of bits in the standard binary representation of $n$, with the convention that the standard binary representation of $0$ is $0$, so $\ell(0)=1$. Since $B_0(n)+B_1(n)=\ell(n)$, we have $B_0(n)\equiv \ell(n)+B_1(n)\pmod 2$. Therefore, to decide whether $B_0(H(x))$ is odd, it is enough to determine the parity of $B_1(H(x))$ and the parity of $\ell(H(x))$.

Define $r=2q+1$ and $G(x)=(2^r-1)H(x)$. Also, for each $i\in\{0,1,\dots,q\}$, define $H_i(x)=H(x)+2^i$. The functions $G(x)$, $H(x)$, and $H_i(x)$ are all ${\bf \#P}$ functions. The function $G(x)$ is obtained by multiplying each accepting path counted by $H(x)$ by $2^r-1$ accepting choices, and each $H_i(x)$ is the sum of $H(x)$ and the length-dependent constant $2^i$. We next define function $F$, where $F(\langle 0,x\rangle)=G(x)$, $F(\langle 1,x\rangle)=H(x)$, and $F(\langle 2,x,i\rangle)=H_i(x)$ for each $i\in{0,1,\ldots,q}$. $F$ outputs $0$ on all other inputs. Since the tag and the index $i$ can be decoded in polynomial time, and since $G(x)$, $H(x)$, and each $H_i(x)$ are ${\bf \# P}$ functions, we have $F\in{\bf \#P}$. Therefore, the single oracle language $O=\{w\mid B_1(F(w))\not\equiv 0\pmod 2\}$ belongs to ${\bf B}_{|1|\oplus}{\bf P}$. The oracle machine can consequently query $G(x)$ using $\langle0,x\rangle$, query $H(x)$ using $\langle1,x\rangle$, and query $H_i(x)$ using $\langle2,x,i\rangle$.

Now the deterministic polynomial time oracle machine works as follows. On input $x$, it first queries the oracle for $G(x)$. If $H(x)=0$, then $G(x)=0$, so $B_1(G(x))=0$, and the oracle returns no. If $H(x)>0$, then $1\le H(x)<2^r$, and by the complementary block argument, $(2^r-1)H(x)$ has exactly $r$ many $1$'s in its binary representation. Since $r=2q+1$ is odd, the oracle returns yes. Thus, the first query returns no exactly when $H(x)=0$, and returns yes exactly when $H(x)>0$. If the first query returns no, the machine accepts immediately, because then $H(x)=0$ and $B_0(0)=1$ is odd. If the first query returns yes, then $H(x)>0$. The machine queries the oracle for $H(x)$ itself, obtaining the parity of $B_1(H(x))$. It remains to determine the parity of $\ell(H(x))$. Let $h$ be the index of the most significant $1$-bit of $H(x)$, so $\ell(H(x))=h+1$. Since $H(x)\le 2^q$, we have $h\le q$.

To find $h$, the machine uses the queries for $H_i(x)=H(x)+2^i$. Let $c_i$ be the number of consecutive $1$'s in the binary representation of $H(x)$ starting at position $i$ and moving upward; if the bit in position $i$ is $0$, then $c_i=0$. Adding $2^i$ changes those $c_i$ consecutive $1$'s to $0$'s and changes the next $0$ to a $1$, so $B_1(H_i(x))\equiv B_1(H(x))+1+c_i\pmod 2$. Thus the oracle queries for $H_i(x)$ allow the machine to compute $c_i\pmod 2$ for every $i$.

The machine checks the indices in decreasing order, $i=q,q-1,\dots,0$, and stops at the first index for which $c_i$ is odd. This index is exactly $h$, the position of the most significant $1$'s bit of $H(x)$, because for every $i>h$ the bit in position $i$ is $0$, so $c_i=0$, while at $i=h$ the bit is $1$ and the next higher bit is $0$, so $c_h=1$. Thus the machine correctly finds $h$ and computes $\ell(H(x))=h+1$. It already knows $B_1(H(x))\pmod 2$, so using $B_0(H(x))\equiv \ell(H(x))+B_1(H(x))\pmod 2$, it computes $B_0(H(x))\pmod 2$ and accepts exactly when this parity is $1$.

We next prove the main theorem with the following two lemmas.

\begin{lemma}
If $x\in L$, then the oracle machine accepts $x$.
\end{lemma}
\begin{proof}
Assume $x\in L$. Since $L\in{\bf B_{|0|\oplus}P}$ via $H(x)$, we have $B_0(H(x))\not\equiv 0\pmod 2$. Thus $B_0(H(x))$ is odd. 

Suppose $H(x)=0$. The first oracle query is the query for $G(x)=(2^r-1)H(x)$. Since $H(x)=0$, we have $G(x)=0$. Hence $B_1(G(x))=B_1(0)=0$, which is even, so the ${\bf B_{|1|\oplus}P}$ oracle returns no on this first query. By construction, when the first query for $G(x)$ returns no, the oracle machine concludes that $H(x)=0$ and accepts immediately. 

Suppose $H(x)>0$. Then the first oracle query for $G(x)$ returns yes, because $G(x)=(2^r-1)H(x)$ has exactly $r$ many $1$'s in its binary representation and $r=2q+1$ is odd. Therefore, the oracle machine continues. It then queries the oracle for $H(x)$ itself to obtain the parity of $B_1(H(x))$, and it uses the queries for $H_i(x)=H(x)+2^i$ to compute the true binary length $\ell(H(x))$. The machine then computes the parity of $B_0(H(x))$ by using the identity $B_0(H(x))\equiv \ell(H(x))+B_1(H(x))\pmod 2$. Since $x\in L$ and $L$ is recognized by the condition $B_0(H(x))\not\equiv 0\pmod 2$, the number $B_0(H(x))$ is odd. Therefore, the computed parity is $1$, and by construction the oracle machine accepts.$\qed$

\end{proof}

\begin{lemma}
If $x\notin L$, then the oracle machine rejects $x$.
\end{lemma}
\begin{proof}
Assume $x\notin L$. Since $L\in{\bf B_{|0|\oplus}P}$ via $H(x)$, we have $B_0(H(x))\equiv 0\pmod 2$. Thus $B_0(H(x))$ is even. 

Suppose $H(x)>0$. Then the first oracle query is the query for $G(x)=(2^r-1)H(x)$. Since $H(x)>0$, the number $G(x)=(2^r-1)H(x)$ has exactly $r$ many $1$'s in its binary representation. Since $r=2q+1$ is odd, the oracle returns yes on this first query. By construction, when the first query for $G(x)$ returns yes, the oracle machine concludes that $H(x)>0$ and continues. It then queries the oracle for $H(x)$ itself to obtain the parity of $B_1(H(x))$, and it uses the queries for $H_i(x)=H(x)+2^i$ to compute the true binary length $\ell(H(x))$. The machine then computes the parity of $B_0(H(x))$ by using the identity $B_0(H(x))\equiv \ell(H(x))+B_1(H(x))\pmod 2$. Since $x\notin L$ and $L$ is recognized by the condition $B_0(H(x))\not\equiv 0\pmod 2$, the number $B_0(H(x))$ is even. Therefore, the computed parity is $0$, and by construction the oracle machine rejects. $\qed$
\end{proof}

The two lemmas show that the oracle machine decides $L$ using an oracle in ${\bf B_{|1|\oplus}P}$. Therefore, $L\in{\bf P}^{{\bf B_{|1|\oplus}P}}$. Since $L\in{\bf B_{|0|\oplus}P}$ was for any $L$, then we can conclude that ${\bf B_{|0|\oplus}P}\subseteq{\bf P}^{{\bf B_{|1|\oplus}P}}$.$\qed$
\end{proof}

\begin{theorem}
${\bf P}^{{\bf B_{|0|\oplus}P}}={\bf P}^{{\bf B_{|1|\oplus}P}}$
\end{theorem}

\begin{proof}
By the previous result, we have ${\bf B_{|1|\oplus}P}\subseteq {\bf P}^{{\bf B_{|0|\oplus}P}}$. Therefore, every oracle query to a language in ${\bf B_{|1|\oplus}P}$ can be simulated by a deterministic polynomial time computation with access to a ${\bf B_{|0|\oplus}P}$ oracle. Thus, ${\bf P}^{{\bf B_{|1|\oplus}P}}\subseteq {\bf P}^{{\bf P}^{{\bf B_{|0|\oplus}P}}}$. Since nested deterministic polynomial time oracle computations collapse, we have ${\bf P}^{{\bf P}^{{\bf B_{|0|\oplus}P}}}={\bf P}^{{\bf B_{|0|\oplus}P}}$. Therefore, ${\bf P}^{{\bf B_{|1|\oplus}P}}\subseteq {\bf P}^{{\bf B_{|0|\oplus}P}}$.

Similarly, by the previous result, we have ${\bf B_{|0|\oplus}P}\subseteq{\bf P}^{{\bf B_{|1|\oplus}P}}$. Therefore, every oracle query to a language in ${\bf B_{|0|\oplus}P}$ can be simulated by a deterministic polynomial time computation with access to a ${\bf B_{|1|\oplus}P}$ oracle. Thus, ${\bf P}^{{\bf B_{|0|\oplus}P}}\subseteq {\bf P}^{{\bf P}^{{\bf B_{|1|\oplus}P}}}$. Again, nested deterministic polynomial time oracle computations collapse, so ${\bf P}^{{\bf P}^{{\bf B_{|1|\oplus}P}}}={\bf P}^{{\bf B_{|1|\oplus}P}}$. Therefore, ${\bf P}^{{\bf B_{|0|\oplus}P}}\subseteq{\bf P}^{{\bf B_{|1|\oplus}P}}$.

Combining the two containments, we get ${\bf P}^{{\bf B_{|0|\oplus}P}}={\bf P}^{{\bf B_{|1|\oplus}P}}$.
\end{proof}

\begin{theorem}

${\bf P}^{{\bf B_{|0|\oplus}P}}={\bf P}^{{\bf B_{|1|\oplus}P}}\subseteq {\bf P}^{\bf PP}$

\end{theorem}

\begin{proof}
By the previous theorem, we have ${\bf P}^{{\bf B_{|0|\oplus}P}}={\bf P}^{{\bf B_{|1|\oplus}P}}$. It remains to show that this is contained in ${\bf P}^{\bf PP}$. By the previous containment theorem, we have ${\bf B_{|0|\oplus}P}\subseteq{\bf P}^{\bf PP}$. Therefore, any oracle query to a language in ${\bf B_{|0|\oplus}P}$ can be simulated by a deterministic polynomial time computation with access to a ${\bf PP}$ oracle. Thus, ${\bf P}^{{\bf B_{|0|\oplus}P}}\subseteq {\bf P}^{{\bf P}^{\bf PP}}$. Since nested deterministic polynomial time oracle computations collapse, we have ${\bf P}^{{\bf P}^{\bf PP}}={\bf P}^{\bf PP}$.

Therefore, ${\bf P}^{{\bf B_{|0|\oplus}P}}\subseteq{\bf P}^{\bf PP}$. Using ${\bf P}^{{\bf B_{|0|\oplus}P}}={\bf P}^{{\bf B_{|1|\oplus}P}}$, then we can conclude that ${\bf P}^{{\bf B_{|0|\oplus}P}}={\bf P}^{{\bf B_{|1|\oplus}P}}\subseteq {\bf P}^{\bf PP}$.
\end{proof}

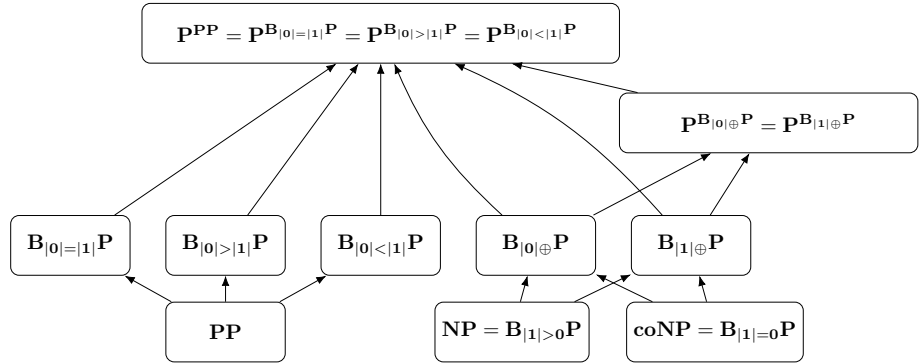
\begin{center}
\resizebox{\textwidth}{!}{
\begin{tikzpicture}[
    every node/.style={font=\normalsize},
    class/.style={draw, rounded corners, minimum width=2cm, minimum height=1cm, align=center},
    oracle/.style={draw, rounded corners, minimum width=5cm, minimum height=1cm, align=center},
    wideoracle/.style={draw, rounded corners, minimum width=8cm, minimum height=1cm, align=center},
    >=Latex
]

\node[wideoracle] (pcomp) at (0,5)
{${\bf P}^{\bf PP} = {\bf P}^{{\bf B_{|0|=|1|}P}}={\bf P}^{{\bf B_{|0|>|1|}P}}={\bf P}^{{\bf B_{|0|<|1|}P}}$
};

\node[oracle] (pparity) at (6.5,3.5)
{${\bf P}^{{\bf B_{|0|\oplus}P}}={\bf P}^{{\bf B_{|1|\oplus}P}}$};

\node[class] (beq)  at (-5.2,1.45) {${\bf B_{|0|=|1|}P}$};
\node[class] (bgt)  at (-2.6,1.45) {${\bf B_{|0|>|1|}P}$};
\node[class] (blt)  at (0,1.45)    {${\bf B_{|0|<|1|}P}$};
\node[class] (b0op) at (2.6,1.45)  {${\bf B_{|0|\oplus}P}$};
\node[class] (b1op) at (5.2,1.45)  {${\bf B_{|1|\oplus}P}$};

\node[class] (pp)   at (-2.6,0) {${\bf PP}$};
\node[class] (np)   at (2.2,0)  {${\bf NP} = {\bf B_{|1|>0}P} $};
\node[class] (conp) at (5.6,0)  {${\bf CoNP} = {\bf B_{|1|=0}P} $};

\draw[->] (pp) -- (beq);
\draw[->] (pp) -- (bgt);
\draw[->] (pp) -- (blt);

\draw[->] (beq) -- (pcomp);
\draw[->] (bgt) -- (pcomp);
\draw[->] (blt) -- (pcomp);

\draw[->] (np) -- (b0op);
\draw[->] (np) -- (b1op);
\draw[->] (conp) -- (b0op);
\draw[->] (conp) -- (b1op);

\draw[->] (b0op) -- (pparity);
\draw[->] (b1op) -- (pparity);

\draw[->] (pparity) -- (pcomp);

\draw[->, bend left=12] (b0op) to (pcomp);
\draw[->, bend right=12] (b1op) to (pcomp);

\end{tikzpicture}

}
\end{center}
Location of our bit-counting complexity classes in the complexity hierarchy. 

\section{Conclusion}

We introduced and studied a family of bit-counting complexity classes defined by applying predicates to the count of 0's bits and the count of 1's bits in the binary expansion of  ${\bf \#P}$ functions. Our first observation was that comparing the count of either the 0's bits or the 1's bits to the value zero results in some well known complexity classes, such as ${\bf B_{|1|>0}P} ={\bf NP}$, ${\bf B_{|1|=0}P} ={\bf CoNP}$, ${\bf B_{|0|=0}P} = {\bf MNS}$ and ${\bf B_{|0|>0}P} = {\bf CoMNS}$. Most interestingly, we also obtained some nontrivial complexity classes and divided them into two subcategories as the comparison based and the parity based bit-counting complexity classes.

The comparison based bit-counting complexity classes ${\bf B_{|0|=|1|}P}$, ${\bf B_{|0|>|1|}P}$ and ${\bf B_{|0|<|1|}P}$ capture whether the number of 0's bits and the number of 1's bits in a ${\bf \#P}$ value are equal, zero dominant, or one dominant. One can also define comparison based bit-counting complexity classes in the flavor of gap definability \cite{FFK94}, such that we can ask whether the difference between the count of 0's bits and the count of 1's bits is equal to zero, greater than zero or less than zero. Nevertheless, we proved that each one of these complexity classes is strong enough to contain the complexity class ${\bf PP}$, that is ${\bf PP}\subseteq{\bf B_{|0|=|1|}P}$, ${\bf PP}\subseteq{\bf B_{|0|>|1|}P}$ and ${\bf PP}\subseteq{\bf B_{|0|<|1|}P}$. We used explicit ${\bf \#P}$ constructions in the proofs that force the binary representation of the resulting count to have the desired balance between 0's bits and 1's bits. Additionally, since a ${\bf P}^{\bf PP}$ machine can recover the exact value of a ${\bf \#P}$ function using threshold queries and then inspect its binary representation directly, each one of the comparison based bit-counting complexity classes is contained in ${\bf P}^{\bf PP}$. By combining these two results, we obtained the following equivalence ${\bf P}^{\bf PP} = {\bf P}^{{\bf B_{|0|=|1|}P}} = {\bf P}^{{\bf B_{|0|>|1|}P}}={\bf P}^{{\bf B_{|0|<|1|}P}}$. In other words, the complexity class ${\bf PP}$ and our comparison based bit-counting complexity classes are equivalent at the Turing level.  

For the parity based bit-counting complexity classes, ${\bf B_{|0|\oplus}P}$ and ${\bf B_{|1|\oplus}P}$, we proved that both of them contain the complexity classes ${\bf NP}$ and ${\bf CoNP}$. These containments show that the parity predicates applied to the number of 0's bits or the number of 1's bits of a ${\bf \#P}$ count are powerful enough to encode existential and universal non-deterministic behavior. At the same time, both classes are contained in ${\bf P}^{\bf PP}$, again because ${\bf P}^{\bf PP}$ can reconstruct the relevant ${\bf \#P}$ count and compute the desired bit parity in polynomial time. We further showed that the two parity based bit-counting complexity classes are equivalent at the Turing level, that is ${\bf P}^{{\bf B_{|0|\oplus}P}}={\bf P}^{{\bf B_{|1|\oplus}P}}$. We did this by proving that when either parity based bit-counting complexity class is provided as an oracle for a polynomial time Turing machine, then it can simulate the other one, that is ${\bf B_{|1| \oplus}P}\subseteq {\bf P}^{{\bf B_{|0| \oplus}P}}$ and ${\bf B_{|0| \oplus}P}\subseteq {\bf P}^{{\bf B_{|1| \oplus}P}}$.

Our results demonstrated that bit-counting complexity classes form a useful framework for studying the relationship between non-deterministic and probabilistic counting complexity. However, there are still more avenues to explore. For instance, we were not able to prove much with regard to the complexity classes ${\bf MNS}$ and ${\bf CoMNS}$, other than the fact that by definition they are equal to ${\bf B_{|0|=0}P}$ and ${\bf B_{|0|>0}P}$, respectively. We were also not able to show how the complexity class ${\bf \oplus P}$ relates to any one of our comparison based or parity based bit-counting complexity classes. This is also true about the complexity class ${\bf US}$, but only with regard to parity based bit-counting complexity classes. Because the complexity class ${\bf US}$ is contained in all of our comparison based bit-counting complexity classes since ${\bf US}\subseteq {\bf C_{=}P}={\bf ES}={\bf MNS}\subseteq {\bf PP}$ and that all three comparison based bit-counting complexity classes contain the complexity class ${\bf PP}$. 

Additionally, can the ${\bf P}^{\bf PP}$ upper bound for the parity-based bit-counting complexity classes be improved? Are complexity classes ${\bf B_{|0|\oplus}P}$ and ${\bf B_{|1|\oplus}P}$ contained in some level of the polynomial hierarchy, ${\bf PH}$ \cite{S76}? Or do complexity classes ${\bf B_{|0|\oplus}P}$ and ${\bf B_{|1|\oplus}P}$ contain some level of {\bf PH}, of course other than {\bf NP} and {\bf CoNP}? It would also be of interest to see if ${\bf B_{|0|\oplus}P} ^ {\bf B_{|0|\oplus}P}$ collapses to ${\bf B_{|0|\oplus}P}$ or if it gets more powerful as more ${\bf B_{|0|\oplus}P}$ oracles are added. Clearly, the same question could be asked about complexity class ${\bf B_{|1|\oplus}P}$. Furthermore, can we show whether ${\bf B_{|1|\oplus}P}^{{\bf B_{|0|\oplus}P}} \subseteq {\bf B_{|0|\oplus}P}^{{\bf B_{|1|\oplus}P}}$ or ${\bf B_{|0|\oplus}P}^{{\bf B_{|1|\oplus}P}} \subseteq {\bf B_{|1|\oplus}P}^{{\bf B_{|0|\oplus}P}}$ hold? We could also investigate four possible parity based bit-counting hierarchy constructions:

-${\bf B_{|1|\oplus}H}$, like in ${\bf PH}$, we use ${\bf B_{|1|\oplus}P}$ and its complement at each level.

-${\bf B_{|0|\oplus}H}$, like in ${\bf PH}$, we use ${\bf B_{|0|\oplus}P}$ and its complement at each level.

-${\bf B_{|1|\oplus}AH}$, we start with ${\bf B_{|1|\oplus}P}$ and then alternate between ${\bf B_{|0|\oplus}P}$  and ${\bf B_{|1|\oplus}P}$ at each level. 

-${\bf B_{|0|\oplus}AH}$, we start with ${\bf B_{|0|\oplus}P}$ and then alternate between ${\bf B_{|1|\oplus}P}$  and ${\bf B_{|0|\oplus}P}$ at each level. 

Do hierarchies ${\bf B_{|1|\oplus}H}$, ${\bf B_{|0|\oplus}H}$, ${\bf B_{|1|\oplus}AH}$ and ${\bf B_{|0|\oplus}AH}$ contain ${\bf PH}$, contain one another, or are equal to each other?

There is also the case of semantic \cite{CP18} variants of our comparison and parity based bit-counting complexity classes. Semantic complexity classes are best conceptualized with the leaf language representation \cite{P94} of complexity classes, where $L_A, L_R\subseteq \{0,1\}^*$ are disjoint leaf languages, with $L_A$ serving as the acceptance criterion and $L_R$ serving as the rejectance criterion. If the leaf language representation of a complexity class has the property that $L_A \cup L_R = \{0,1\}^*$ then it is classified as a syntactic complexity class. On the other hand, if the leaf language representation of a complexity class has the property that $L_A \cup L_R \not = \{0,1\}^*$ then it is classified as a semantic complexity class. Note that in the leaf language representation of a complexity class, each leaf of the computation tree of a non-deterministic polynomial time Turing machine is assigned a 0 if it is a rejecting path and assigned a 1 if it is an accepting path. This is somewhat similar in spirit to the bit-counting complexity classes, but mathematically distinct. Because the bit-counting complexity classes look at the binary representation of the total number of 1's in the leaf language representation and then apply bit-counting predicates to that value. Nonetheless, the complexity classes ${\bf UP}$ \cite{V76}, ${\bf EP}$ \cite{BHR00}, and  ${\bf MNP}$ \cite{CP18} are the semantic variants of complexity classes ${\bf US}$, ${\bf ES}$ and ${\bf MNS}$, respectively. Whereas complexity classes ${\bf Half\_P}$ \cite{BB92} and ${\bf F_{=}P}$ \cite{CP18} are two different semantic variants of the complexity class ${\bf C_=P}$. Similarly, complexity classes ${\bf RP}$\cite{G77} and ${\bf BPP}$\cite{G77} are two different semantic variants of the complexity class ${\bf PP}$. Some other important semantic complexity classes are ${{\bf UP}_{O(k)}}$ \cite{B89} and ${\bf FewP}$\cite{A86}. 

Essentially, we can create semantic variants of our comparison and parity based bit-counting complexity classes by revising their definition so that the rejectance criterion is explicitly forced while the acceptance criterion stays the same. For instance, for bit-counting complexity classes ${\bf B_{|0|=|1|}P}$, ${\bf B_{|0|<|1|}P}$ and ${\bf B_{|1|\oplus}P}$, we can make their rejection criterion zero accepting paths and for bit-counting complexity classes ${\bf B_{|0|>|1|}P}$ and ${\bf B_{|0|\oplus}P}$ we can make their rejection criterion one accepting path. This yields their natural semantic variants and results in the following containments purely based on their respective definitions: 

-${\bf B_{|0|=|1|}SP} \subseteq {\bf B_{|0|=|1|}P}$

-${\bf B_{|0|<|1|}SP} \subseteq {\bf B_{|0|<|1|}P}$

-${\bf B_{|0|>|1|}SP} \subseteq {\bf B_{|0|>|1|}P}$

-${\bf B_{|0|\oplus}SP} \subseteq {\bf B_{|0|\oplus}P}$

-${\bf B_{|1|\oplus}SP} \subseteq {\bf B_{|1|\oplus}P}$

-${\bf B_{|0|=|1|}SP} \subseteq {\bf NP}$

-${\bf B_{|0|<|1|}SP} \subseteq {\bf NP}$

-${\bf B_{|0|>|1|}SP} \subseteq {\bf CoUS}$

-${\bf B_{|0|\oplus}SP} \subseteq {\bf CoUS}$

-${\bf B_{|1|\oplus}SP} \subseteq {\bf NP}$

Do any of the theorems we proved in this paper with regard to our comparison and parity based bit-counting complexity classes also hold for their semantic variants? Can we prove any containment between our semantic bit-counting complexity classes? Can we prove any containment between our semantic  bit-counting complexity classes and any classical complexity class, whether it is semantic or syntactic? 

One last point, every Mersenne prime $2^p-1$ other than 3 has an odd number of 1's bits, since its exponent $p$ is an odd prime. Thus, ${\bf B_{|1|\oplus}P}\cap{\bf B_{|0|=0}P}$ might be of interest for further exploration.

\bibliographystyle{alpha}
\bibliography{Bit_Counting_Complexity_Classes}

\end{document}